\documentclass[conference]{IEEEtran}
\IEEEoverridecommandlockouts
\usepackage[none]{hyphenat}
\usepackage{cite}
\usepackage{graphicx}
\usepackage{amsmath,amssymb,amsthm}
\usepackage{booktabs}
\usepackage{balance}
\usepackage[colorlinks=true,linkcolor=blue,citecolor=blue,urlcolor=blue]{hyperref}
\usepackage{tikz}
\usetikzlibrary{arrows.meta,positioning,calc}
\definecolor{inkblue}{HTML}{225B82}
\definecolor{inkgreen}{HTML}{32745E}
\definecolor{inkorange}{HTML}{AF612C}

\graphicspath{{figures/}}

\newcommand{\Eerr}{\mathcal{E}_{\mathrm{err}}}
\newcommand{\Eobs}{\mathcal{E}_{\mathrm{obs}}}
\newcommand{\Uid}{U_{\mathrm{id}}}
\newcommand{\rca}{r_{\mathrm{cross,anti}}}
\newcommand{\drca}{\Delta r_{\mathrm{cross,anti}}}
\newcommand{\Aedge}{\mathcal{A}^{G}_{\mathrm{edge}}}
\newcommand{\wt}{\mathrm{wt}}
\newcommand{\Tr}{\mathrm{Tr}}

\newtheorem{proposition}{Proposition}
\newtheorem{lemma}{Lemma}

\begin{document}
\sloppy

\title{\LARGE{Distinguishing Coherent Crosstalk from Calibration Drift via Pauli-Transfer Signatures and Quantum Edge Detection}\thanks{Accepted for publication in the 2027 IEEE International Symposium on Hardware Oriented Security and Trust (HOST 2027).}}

\author{\IEEEauthorblockN{Syed Emad Uddin Shubha and Tasnuva Farheen}
\IEEEauthorblockA{Division of Computer Science and Engineering\\
Louisiana State University, Baton Rouge, LA, USA}}

\maketitle
\begin{abstract}
Multi-tenant quantum processors expose a pulse-level attack surface in which coherent crosstalk can mimic benign calibration drift at matched average gate infidelity. We present a structural verification and detection framework based on the residual Pauli transfer matrix (PTM). We prove that products of local unital, trace-preserving channels cannot mix weight-1 and weight-2 Pauli operators, even under coherent drift of arbitrary strength or axis. Building on established PTM-to-Hamiltonian relations, we prove that, for any specified qubit pair in an arbitrarily large system, the first-order map from its nine interaction coefficients to the antisymmetric cross-weight feature is an isometry up to scale. Thus, every interaction direction is locally observable near the identity. We also exhibit cancellation of this feature under large local rotations and prove that the full cross-weight norm is invariant under local unitary composition, providing a structural countermeasure. For calibrated small local rotations, shared randomized-Pauli measurements support periodic verification using the evaluated antisymmetric detector. In simulation, $16{,}384$ randomized settings yield a detection threshold of $\lambda_{\min}\approx0.13$ at a $5\%$ calibrated false-positive target, with approximately $M^{-1/2}$ scaling in the number of settings $M$. Further experiments characterize detection under relaxation, depolarization, native ZZ fluctuations, and time-dependent pulse dynamics. Finally, an FRQI encoding with Pauli-graph quantum Hadamard edge detection reproduces the classical structural edge score to numerical precision.
\end{abstract}

\begin{IEEEkeywords}
Quantum hardware security; Coherent crosstalk; Calibration drift; Pauli transfer matrix; Randomized-Pauli shadows; Edge detection; FRQI
\end{IEEEkeywords}

\section{Introduction}
\label{sec:introduction}

Quantum processors are increasingly offered as shared, cloud-scheduled
systems. In these environments, several users may access nearby qubits, and
pulse-level control can allow disturbances on one qubit to affect another.
Such disturbances may arise from cross-resonance leakage,
spectator-conditioned drives, imperfect echo cancellation, or unwanted
always-on coupling
\cite{sarovar2020crosstalk,shubha2025stealthy, oda2026sparse,rudinger2021sgst}.
These effects create a hardware-security concern because a weak coherent
interaction can be introduced during a victim gate without producing an
obvious failure.

The main challenge is that coherent crosstalk can resemble ordinary
calibration drift. A local over-rotation and a two-qubit interaction can be
chosen to produce the same average gate infidelity. Standard scalar
diagnostics, including average fidelity and randomized benchmarking error,
can then reveal that the gate is inaccurate but cannot determine whether the
cause is local drift or a nonlocal interaction. This ambiguity matters because
the two cases require different responses. Local drift may require
recalibration, while an unexpected two-qubit interaction may indicate control
leakage, faulty isolation, or an injected pulse-level disturbance.

Prior work has studied the characterization, mitigation, and calibration of
crosstalk in superconducting processors
\cite{sarovar2020crosstalk,oda2026sparse}, and
off-diagonal PTM entries have been related to local and two-qubit
Hamiltonian coefficients \cite{kaufmann2025coherent} and estimated from randomized product-Pauli data \cite{crupi2025lowdegree}. Recent security studies have also
shown that structured pulse-level perturbations can alter a victim
computation while remaining difficult to identify through conventional
performance metrics
\cite{shubha2025stealthy}. These results do not, however,
provide a security criterion for separating local coherent drift
from coherent crosstalk when their error magnitudes are matched; we build on
these characterization tools to formulate one. The Pauli transfer matrix describes how a channel transforms one Pauli
operator into another \cite{korotkov2013error,hantzko2024ptm}, and
randomized-Pauli and classical-shadow methods estimate many of its entries
from one shared dataset \cite{huang2020shadows,levy2021shadowqpt}. The
remaining question is which part of it reliably distinguishes local drift
from nonlocal coherent interaction.

We show that the relevant information lies in the signed antisymmetric
coupling between weight-1 and weight-2 Pauli sectors. Products of local
unital, trace-preserving errors, including coherent drift of arbitrary
strength, axis, and location, do not create such mixing, whereas a genuine
two-qubit coherent interaction does at first order. This defines a signed
Pauli-transfer signature for direction-resolved detection relative to an
attack-free calibration baseline.

We also study a quantum-native readout. Because neighboring PTM indices
need not represent related Pauli operators, we define a Pauli graph whose
edges connect labels that differ in one tensor factor, encode the
PTM-derived magnitude pattern as a flexible representation of quantum
images (FRQI) state \cite{le2011frqi}, and measure its variation over this
graph with quantum Hadamard edge detection (QHED)
\cite{yao2017qhed,shubha2024edge}. The signed PTM vector remains the
detector. \\
\noindent
This work addresses the following questions:

\begin{enumerate}
    \item[\textbf{RQ1.}]
    Can coherent two-qubit crosstalk be distinguished from local calibration
    drift when both produce the same average gate infidelity?

    \item[\textbf{RQ2.}]
    Does the proposed cross-weight Pauli-transfer signature capture every
    two-qubit Pauli interaction direction, or can some interactions remain
    invisible?

    \item[\textbf{RQ3.}]
    Can this signature be estimated with finite randomized measurements and
    remain effective under relaxation, depolarizing noise, native coupling
    fluctuations, and time-dependent pulse dynamics?

    \item[\textbf{RQ4.}]
    Can the resulting PTM feature be represented and measured through an
    FRQI and Pauli-graph quantum edge-detection circuit?
\end{enumerate}

\noindent \textbf{Contributions:}
In answer to these questions, this paper makes four contributions, and we
distinguish their levels of evidence explicitly. First, it proves an exact
algebraic result: the weight-1 to weight-2 block of the residual PTM is
identically zero for every product of local unital, trace-preserving error
channels, and any two-qubit Pauli generator populates this block at first
order (Section~\ref{sec:signature}). Second, building on the established
first-order PTM-to-Hamiltonian relation \cite{kaufmann2025coherent}, it
proves that the map from the nine interaction coefficients into this block
is an isometry up to scale for every system size, and it gives a full-block
invariant for large local rotations that cancel the antisymmetric
projection.
Third, it specifies a single operational detector, estimated from one shared
randomized-Pauli measurement dataset and calibrated on attack-free data, and
evaluates it empirically under realistic noise, finite measurement budgets,
fluctuating native ZZ coupling, and time-ordered pulse dynamics
(Sections~\ref{sec:detector} and~\ref{sec:evaluation}). Fourth, it
constructs a Pauli-graph QHED circuit that measures the structural pattern
of the PTM feature from an FRQI encoding and verifies the circuit against
the classical score (Section~\ref{sec:qhed}).

The PTM discriminator and the calibrated hypothesis test are the primary
security contribution; the FRQI and QHED construction is a secondary
representation result that improves neither detection accuracy nor
computational complexity. Figure~\ref{fig:overview} summarizes the
pipeline.

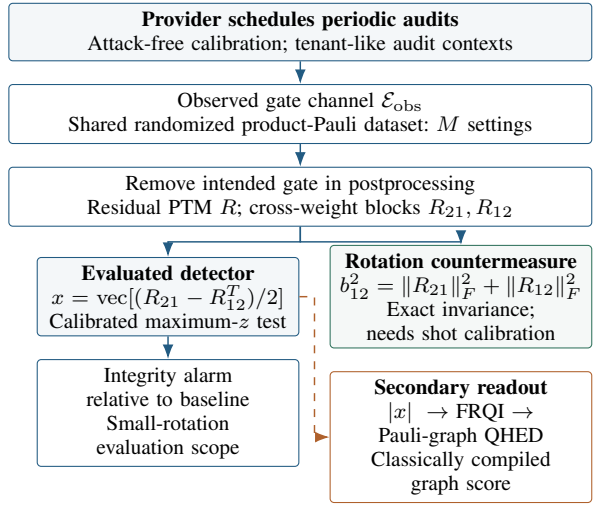
\begin{figure}[tb]
\centering
\begin{tikzpicture}[
  font=\footnotesize, >=Latex,
  box/.style={draw=inkblue,rounded corners=1.5pt,align=center,
    inner sep=3pt,text width=7.5cm},
  branch/.style={box,text width=3.25cm,minimum height=0.9cm},
  line/.style={->,line width=.55pt,draw=inkblue}]
\node[box,fill=inkblue!5] (model)
 {\textbf{Provider schedules periodic audits}\\
 Attack-free calibration; tenant-like audit contexts};
\node[box,below=3mm of model] (data)
 {Observed gate channel $\mathcal E_{\rm obs}$\\
 Shared randomized product-Pauli dataset: $M$ settings};
\node[box,below=3mm of data] (ptm)
 {Remove intended gate in postprocessing\\
 Residual PTM $R$; cross-weight blocks $R_{21},R_{12}$};
\node[branch,below left=4mm and -3.85cm of ptm,fill=inkblue!5] (primary)
 {\textbf{Evaluated detector}\\
 $x=\operatorname{vec}[(R_{21}-R_{12}^{T})/2]$\\
 Calibrated maximum-$z$ test};
\node[branch,right=4mm of primary,draw=inkgreen,fill=inkgreen!5] (counter)
 {\textbf{Rotation countermeasure}\\
 $b_{12}^2=\|R_{21}\|_F^2+\|R_{12}\|_F^2$\\
 Exact invariance; needs shot calibration};
\node[branch,below=3mm of primary] (decision)
 {Integrity alarm relative to baseline\\
 Small-rotation evaluation scope};
\node[branch,below=3mm of counter,draw=inkorange] (qhed)
 {\textbf{Secondary readout}\\
 $|x|\to$ FRQI $\to$ Pauli-graph QHED\\
 Classically compiled graph score};
\draw[line] (model)--(data);
\draw[line] (data)--(ptm);
\draw[line] (ptm.south) -- ++(0,-2mm) -| (primary.north);
\draw[line] (ptm.south) -- ++(0,-2mm) -| (counter.north);
\draw[line] (primary)--(decision);
\draw[line,dashed,draw=inkorange] (primary.east) -- ++(2mm,0)
  |- (qhed.west);
\end{tikzpicture}
\caption{Overview of the framework. One shared randomized-Pauli dataset
yields the residual PTM and the signed cross-weight vector $x$, which is
exactly zero for product local unital errors
(Proposition~\ref{prop:drift-zero}) and first order in any two-qubit Pauli
generator (Proposition~\ref{prop:attack-first-order}). The calibrated
direction-resolved test on $x$ carries the security decision within the
small-rotation scope; the full-block norm $b_{12}$ is invariant under local
rotations (Proposition~\ref{prop:full}), and the FRQI and
Pauli-graph QHED branch is a quantum-native readout of the
feature's structural pattern.}
\label{fig:overview}
\end{figure}

The method detects unexpected nonlocal coherent structure relative to a
calibrated baseline; it does not infer malicious intent, and all detection
results are hardware-motivated simulations at three qubits.
Section~\ref{sec:limitations} states the detection boundary.

\section{Background and Related Work}
\label{sec:related}

\subsection{Crosstalk characterization}

Crosstalk is a dominant error source in superconducting processors, and a
substantial literature addresses its detection and modeling. Sarovar et
al.~\cite{sarovar2020crosstalk} formalize crosstalk errors as violations of
locality and independence assumptions and propose protocols that detect
their presence.   Simultaneous gate set tomography distinguishes crosstalk mechanisms
  experimentally on trapped-ion and transmon hardware \cite{rudinger2021sgst}.
  Recent characterization work shows that benign transmon
hardware carries persistent parasitic ZZ coupling and two-level-system
mediated interactions that can be learned with a few parameters per qubit
pair \cite{oda2026sparse}. These studies target device bring-up and
calibration rather than the security question of whether an observed
coherent error is local or nonlocal at matched magnitude, which is the gap
this paper fills.

\subsection{Pulse-level attacks on shared hardware}

Security analyses of multi-tenant quantum computing have demonstrated that
pulse-level access enables crosstalk-based attacks. Prior work simulates
cross-resonance pulse injections that corrupt a victim circuit and characterizes stealthy variants that remain
below conventional detection thresholds \cite{shubha2025stealthy}. These
studies establish the attack surface but leave open how a defender can
attribute an observed error to an injected interaction rather than to
drift. As a countermeasure, that line of work proposes canary circuits: simple probe circuits executed periodically on idle qubits and flagged when their output deviates from a pre-calibrated baseline \cite{shubha2025stealthy}. Canaries inherit the scalar limitation, since a deviation shows that the device has changed but not whether the cause is drift or an injected interaction, and an attack matched to the calibration budget can hide inside the baseline tolerance. The protocol in this paper can be read as a structural canary that addresses this gap: the randomized-Pauli audit circuits play the canary role, while the cross-weight signature replaces the scalar baseline comparison with a mechanism-resolved and direction-resolved test.

\subsection{Channel estimation and coherence of noise}

The Pauli transfer matrix is a standard representation of quantum channels
\cite{korotkov2013error,hantzko2024ptm}. Classical-shadow and
randomized-Pauli methods estimate many channel properties from one shared
dataset with cost that depends on operator weight rather than system size
\cite{huang2020shadows,levy2021shadowqpt}. Kaufmann et
al.~\cite{kaufmann2025coherent} relate off-diagonal PTM entries to all
single- and two-qubit coherent rotation coefficients and estimate them on
superconducting hardware; our $C$ and $L$ (Section~\ref{sec:signature}) use
that relation. Wallman et
al.~\cite{wallman2015coherence} estimate how coherent a noise process is,
which is a scalar property. Our contribution is not a new estimation
protocol but a security test on a specific PTM block whose occupancy
separates local from nonlocal coherent errors, which no scalar coherence
measure can do.

\subsection{Quantum image processing}

FRQI encodes an image into rotation angles of a color qubit entangled with
an address register \cite{le2011frqi}, and QHED measures intensity
differences between neighboring pixels with a Hadamard transform
\cite{yao2017qhed,shubha2024edge}. We replace pixel adjacency with a
Pauli-graph adjacency for operator-indexed channel data.

\section{Threat Model}
\label{sec:threat}

\subsection{System model}

Consider an $n$-qubit processor containing a victim qubit $v$ and one or
more spectator qubits. The intended operation on the monitored subsystem is
the unitary channel $\mathcal{U}_{\mathrm{id}}(\rho)=\Uid\rho\Uid^\dagger$.
The actual implementation is a noisy channel $\Eobs$. The defender wishes to
determine whether the difference between $\Eobs$ and
$\mathcal{U}_{\mathrm{id}}$ is consistent with the calibrated device or
contains an additional coherent two-qubit interaction.

\subsection{Calibrated null model}

Let $\mathfrak{N}$ denote the family of channels expected during normal
device operation: depolarizing noise, dephasing, amplitude damping, thermal
relaxation, local coherent miscalibration of the form
$U_{\mathrm{loc}}=\bigotimes_q e^{-i\epsilon_q P_q/2}$, and calibrated
always-on parasitic coupling. The calibrated hardware baseline may contain
persistent two-qubit ZZ coupling and TLS-mediated interactions; such terms
can be estimated using pairwise characterization methods such as JAZZ-type
amplification experiments \cite{oda2026sparse}. Under the alternative
hypothesis below, local rotations composed with an attack satisfy
$|\epsilon_q|\le\epsilon_{\max}$ with
$2[\sin(\epsilon_{\max}/2)+\sin\epsilon_{\max}]<1$; $\epsilon_{\max}=0.3$ rad
suffices and exceeds the $0.15$ rad drift of the finite-shot studies.

The exact structural-zero result of Section~\ref{sec:signature} applies to
products of local unital, trace-preserving channels. The complete hardware
null need not satisfy this product assumption, because amplitude damping is
nonunital and native ZZ is a two-qubit interaction. Operational detection is
therefore performed relative to the measured calibration distribution, not
relative to the mathematical zero alone.

\subsection{Attack model}

The monitored attack family consists of coherent two-qubit interactions
applied during the victim gate. A general interaction on spectator $s$ and
victim $v$ is
\begin{equation}
H_{\mathrm{int}}
=
\sum_{a,b\in\{X,Y,Z\}}
h_{ab}\,
\sigma_a^{(s)}
\otimes
\sigma_b^{(v)},
\label{eq:general-interaction}
\end{equation}
where the nine real coefficients $h_{ab}$ specify the interaction direction.
The corresponding unitary attack is
$U_{\mathcal{A}}(\lambda)=e^{-i\lambda H_{\mathrm{int}}/2}$ with channel
$\mathcal{A}_\lambda(\rho)=U_{\mathcal{A}}\rho U_{\mathcal{A}}^\dagger$.
Here $\lambda$ is the dimensionless interaction angle accumulated over the
monitored gate: for a constant physical coupling of strength $J$ entering
the Hamiltonian as $J H_{\mathrm{int}}/2$ over gate duration
$t_g$, $\lambda=Jt_g$, and for the time-dependent envelopes of
Section~\ref{sec:pulse} it is the corresponding integrated angle. The
evaluated attack families include ZX, YX, and excess ZZ interactions,
together with pulse-level realizations such as detuned cross-resonance
leakage and spectator-conditioned drives.

Against an adaptive adversary who knows the detector, this class is
motivated by pulse-level cross-resonance access, which naturally produces
coherent spectator-conditioned terms that change the channel at first order
\cite{shubha2025stealthy}. Stochastic correlated attacks evade the
antisymmetric signature (Section~\ref{sec:limitations}), and excess ZZ
hidden inside the native fluctuation band is capped at that band
(Section~\ref{sec:nativezz}).

\subsection{Defender capability}

The defender can schedule diagnostic circuits containing the victim gate and
collect randomized Pauli preparation and measurement data under controllable
conditions
$c\in\{\text{idle},\text{driven},\text{calibration time},\text{observation
time}\}$. Scheduling audits and controlling the spectator condition require
provider-side access, so the primary deployment is a provider-run integrity
monitor. A tenant without scheduling privileges can still run the passive
single-condition variant on its own gates, which retains the structural
discriminator but not the context-contrast handle. The deployment assumes
an attack-free calibration window and audits that the adversary cannot
distinguish from tenant jobs; an adversary that recognizes audit windows and
stays silent evades observation (Section~\ref{sec:limitations}).

No privileged access to analog control electronics is required: Clifford
diagnostic gates reduce removal of the intended unitary to classical Pauli
relabeling, and for a non-Clifford gate the residual PTM remains estimable
by linearity at a post-processing cost set by the number of resulting Pauli
terms.

\subsection{Hypothesis test}

The residual error channel is
$\Eerr=\mathcal{U}_{\mathrm{id}}^\dagger\circ\Eobs$. The defender tests
\begin{equation}
\begin{aligned}
H_0&:\ \Eerr\in\mathfrak{N},\\
H_1&:\ \Eerr=\mathcal{D}\circ\mathcal{M}_{\epsilon'}\circ
\mathcal{A}_\lambda\circ\mathcal{M}_{\epsilon},\\
&\quad|\epsilon_q|,|\epsilon'_q|\le\epsilon_{\max},\quad
\lambda>\lambda_{\min},
\end{aligned}
\label{eq:hypothesis-test}
\end{equation}
at a controlled false-positive rate, where $\mathcal{M}_\epsilon$ and
$\mathcal{M}_{\epsilon'}$ are local rotations before and after the
interaction and $\mathcal{D}$ is the remaining calibrated noise in
$\mathfrak{N}$. The angular bound is necessary: a large local rotation, such
as a spectator $\pi$ rotation, can cancel the antisymmetric signature
(Section~\ref{sec:rotation}). Thresholds are selected from held-out
calibration data, and Table~\ref{tab:surface} states which attack classes
each handle addresses.

\begin{table}[tb]
\centering
\caption{Detectability boundary of the proposed protocol.
$\checkmark$ indicates direct sensitivity,
$\times$ indicates that the attack can evade the corresponding test, and
$\circ$ indicates partial sensitivity. Cross-weight refers to the evaluated
antisymmetric statistic $x$; the full-block norm of
Proposition~\ref{prop:full} restores sensitivity in the $\pi$-rotation row.}
\label{tab:surface}
\setlength{\tabcolsep}{3.5pt}
\begin{tabular}{lccc}
\toprule
Attack class & Cross-weight & Temporal & Context \\
\midrule
Coherent two-qubit interaction & $\checkmark$ & $\checkmark$ & $\checkmark$ \\
Interaction $+$ local $\pi$ rotation & $\times$ & $\circ$ & $\circ$ \\
Local coherent drift & $0$ (Prop.~\ref{prop:drift-zero}) & $\circ$ & $0$ \\
Incoherent correlated noise & $\times$ & $\circ$ & $\circ$ \\
Coherent interaction of weight $>k$ & $\times$ & $\circ$ & $\circ$ \\
Excess ZZ inside native band & $\circ$ & $\circ$ & $\circ$ \\
Attack present during calibration & $\times$ & $\times$ & $\circ$ \\
\bottomrule
\end{tabular}
\end{table}

\section{The Structural Pauli-Transfer Signature}
\label{sec:signature}

\subsection{Pauli transfer representation}

The $n$-qubit Pauli basis is
$\mathcal{P}_n=\{P_\alpha=\bigotimes_q P_{\alpha_q}:
P_{\alpha_q}\in\{I,X,Y,Z\}\}$ with
$\Tr(P_\alpha P_\beta)=2^n\delta_{\alpha\beta}$. The weight $\wt(P_\alpha)$
counts the non-identity tensor factors; for example $\wt(IIX)=1$ and
$\wt(IZX)=2$. We monitor the low-weight sector
$\mathcal{L}_k=\{P:\wt(P)\le k\}$ of size
$K=\sum_{w=0}^{k}3^w\binom{n}{w}$; throughout the numerics $n=3$ and $k=2$,
so $K=37$. For fixed $k$, $K$ grows polynomially in $n$, and the low-weight
block is estimable by classical shadows at polynomial cost
\cite{huang2020shadows,levy2021shadowqpt}.

The PTM of the residual channel is
\begin{equation}
R_{\alpha\beta}
=
2^{-n}\,
\Tr\!\left[P_\alpha\,\Eerr(P_\beta)\right],
\label{eq:ptm}
\end{equation}
so that $\Eerr(P_\beta)=\sum_\alpha R_{\alpha\beta}P_\alpha$. The PTM is
real for physical channels, equals the identity for a perfect gate, and has
an orthogonal traceless part for unitary channels
\cite{korotkov2013error,hantzko2024ptm}.

For a coherent error $U_\lambda=e^{-i\lambda G/2}$ generated by a Hermitian
$G$, expanding
$U_\lambda P_\beta U_\lambda^\dagger
=P_\beta-\tfrac{i\lambda}{2}[G,P_\beta]+O(\lambda^2)$
gives $R=I+\lambda C+O(\lambda^2)$, where
$C_{\alpha\beta}=2^{-n}\Tr[P_\alpha(-\tfrac{i}{2})[G,P_\beta]]$ satisfies
$C^T=-C$. Coherent errors therefore appear, at first order, in the
antisymmetric part of the PTM. This motivates the signed component
\begin{equation}
A_{\alpha\beta}
=
\tfrac{1}{2}
\left(
R_{\alpha\beta}-R_{\beta\alpha}
\right),
\label{eq:signed-antisymmetric}
\end{equation}
which retains the sign of each coherent Pauli transition and hence the
direction of the underlying Hamiltonian interaction.

\subsection{The cross-weight signature vector}

Define the oriented index set
$\mathcal{I}_{12}=\{(\alpha,\beta):\wt(P_\alpha)=2,\ \wt(P_\beta)=1\}$ and
the signed cross-weight block $A^{12}_{\alpha\beta}=A_{\alpha\beta}$ for
$(\alpha,\beta)\in\mathcal{I}_{12}$. The structural signature is the vector
\begin{equation}
x
=
\operatorname{vec}\!\left(A^{12}\right)
\in\mathbb{R}^{K_1K_2},
\label{eq:x-definition}
\end{equation}
where $K_1=3n$ and $K_2=9\binom{n}{2}$ count the weight-1 and weight-2
Pauli strings. For $n=3$, $K_1=9$, $K_2=27$, and $x\in\mathbb{R}^{243}$.
Each coordinate of $x$ is one signed coherent transition from a weight-1
input Pauli to a weight-2 output Pauli.

Both ingredients of this definition are essential. Restricting to
cross-weight entries excludes local unital product errors
(Proposition~\ref{prop:drift-zero} below), while antisymmetrization selects
the first-order coherent component and suppresses symmetric stochastic
mixing; unlike the cross-weight restriction, it is not invariant under large
local rotations (Section~\ref{sec:rotation}). Nonunital product noise, such as amplitude damping, can still
create raw cross-weight leakage; this residual contribution is handled
operationally through baseline calibration
(Sections~\ref{sec:detector} and~\ref{sec:separation}).

\subsection{Structural zero for local unital errors}

\begin{proposition}[Local errors have zero cross-weight mass]
\label{prop:drift-zero}
Let $\mathcal{E}=\bigotimes_{q}\mathcal{E}_q$ be any tensor product of
single-qubit channels, each trace preserving and unital. Then every PTM
entry with $\wt(P_\alpha)\neq\wt(P_\beta)$ vanishes, and in particular
$x=0$. This holds for local coherent drift $e^{-i\epsilon P_q/2}$ of any
strength $\epsilon$, on any qubit, about any axis.
\end{proposition}

\begin{proof}
The PTM of a product channel factorizes as
$R_{\alpha\beta}=\prod_q R^{(q)}_{\alpha_q\beta_q}$. If
$\wt(P_\alpha)\neq\wt(P_\beta)$, there is a site $q$ where exactly one of
$\alpha_q,\beta_q$ is the identity. At that site the factor is either
$R^{(q)}_{I,P}=\tfrac12\Tr[\mathcal{E}_q(P)]=0$, because the channel is
trace preserving and $P$ is traceless, or
$R^{(q)}_{P,I}=\tfrac12\Tr[P\,\mathcal{E}_q(I)]=0$, because the channel is
unital. Either way the product vanishes.
\end{proof}

This result is exact rather than perturbative: local coherent drift is not
merely small in the monitored block, it is algebraically forbidden from
entering it.

\subsection{First-order response to two-qubit interactions}

\begin{proposition}[Two-qubit generators are visible at first order]
\label{prop:attack-first-order}
Let $\Eerr(\rho)=e^{-i\lambda G/2}\rho\,e^{i\lambda G/2}$ with
$G=P_s\otimes P_v$ a two-qubit Pauli. Then $\|x\|_2=\Theta(|\lambda|)$ as
$\lambda\to0$.
\end{proposition}

\begin{proof}[Proof sketch]
Choose a weight-1 Pauli $Q$ on the victim with $\{Q,P_v\}=0$. Then
$[G,Q]=P_s\otimes[P_v,Q]$ is proportional to a weight-2 Pauli, giving a
nonzero entry of the antisymmetric generator $C$ in the weight-1 to
weight-2 block, hence mass proportional to $\lambda$ at first order.
\end{proof}

As a concrete example, take $G=Z_sX_v$ and input $P_\beta=I_sY_v$. Since
$[Z_sX_v,I_sY_v]=2iZ_sZ_v$, conjugation gives
$\cos(\lambda)\,I_sY_v+\sin(\lambda)\,Z_sZ_v$: the weight-1 input acquires a
weight-2 component whose PTM entry scales as
$\sin(\lambda)=\lambda+O(\lambda^3)$.

\subsection{No monitored interaction direction is invisible}
\label{sec:isometry}

Proposition~\ref{prop:attack-first-order} covers a single Pauli product,
but a physical interaction is the general Hamiltonian of
Eq.~\eqref{eq:general-interaction}. Collecting the nine coefficients into
$h\in\mathbb{R}^9$, the first-order signature is linear,
\begin{equation}
x
=
\lambda\,Lh
+
O(\lambda^2),
\label{eq:linear-map}
\end{equation}
with $L:\mathbb{R}^9\to\mathbb{R}^{K_1K_2}$ determined by the Pauli
commutators. Two failure modes would undermine detection: a nonzero
interaction in the null space of $L$ would be invisible, and two interaction
directions could cancel.

Neither occurs, for any system size.

\begin{lemma}[Pair-local isometry]\label{lem:isometry}
For any $n\ge2$, $k\ge2$, and any fixed pair $(s,v)$, $L^{T}L=4I_9$. Thus
$\operatorname{rank}(L)=9$, all nine singular values equal $2$, and the
nonzero entries of $L$ are $36$ values $\pm1$ in rows whose Paulis are
supported on the pair.
\end{lemma}
\begin{proof}
For $G_{ab}=\sigma_a^{(s)}\sigma_b^{(v)}$, exactly four weight-1 inputs
anticommute with $G_{ab}$, two on each qubit of the pair, and the
commutator maps each to $\pm1$ times a weight-2 Pauli on the pair; inputs on
other qubits commute with $G_{ab}$. Each input-output coordinate identifies
$(a,b)$, so distinct columns have disjoint supports of squared norm $4$. The
normalized trace in $C$ removes identity factors on other qubits, so larger
registers only add zero rows.
\end{proof}

The same construction applied to the
six local generators $\sigma_a^{(s)}\otimes I$ and $I\otimes\sigma_b^{(v)}$
yields exactly zero columns, consistent with
Proposition~\ref{prop:drift-zero}, and Pauli-algebra checks at
$n=2,\dots,5$ reproduce the lemma exactly. Consequently, every two-qubit
interaction direction produces cross-weight mass,
$\|x\|_2=2|\lambda|\,\|h\|_2+O(\lambda^2)$, no two directions cancel, and
$x$ faithfully encodes the nine-parameter interaction tensor near the
identity; Section~\ref{sec:rotation} treats composition with local
rotations. This answers RQ2 and
justifies the direction-resolved detector of Section~\ref{sec:detector}:
distinct interactions map to distinct, non-cancelling directions in $x$, so
a per-coordinate statistic can separate an attack direction from a different
native-coupling direction.

\subsection{Local-rotation cancellation and full-block invariance}
\label{sec:rotation}

Write $R_{rs}$ for the block of $R$ with weight-$r$ rows and weight-$s$
columns, so that $x=\operatorname{vec}(A^{-}_{21})$ with
$A^{\mp}_{21}=(R_{21}\mp R_{12}^{T})/2$. For $V_\epsilon=e^{-i\epsilon X_s/2}$
composed in either order with $U_{\mathcal A}=e^{-i\lambda Z_sX_v/2}$,
\begin{equation}
\|A^{-}_{21}\|_F=2|\sin\lambda\cos\tfrac{\epsilon}{2}|,\quad
\|A^{+}_{21}\|_F=2|\sin\lambda\sin\tfrac{\epsilon}{2}|,
\label{eq:cancellation}
\end{equation}
and the same holds for YX and ZZ. At $\epsilon=\pi$, $x=0$ for every
$\lambda$, although the channel remains nonlocal when $\sin\lambda\ne0$.
Halving $\|x\|_2$ already requires $|\epsilon|\ge2\pi/3$, a spectator error
with single-qubit average infidelity $(1-\cos\epsilon)/3\ge0.5$ that is
visible in the same-weight block $R_{11}$ and in routine calibration.

\begin{proposition}[Full-block invariance]\label{prop:full}
Let $b_{12}(R)^2=\|R_{21}\|_F^2+\|R_{12}\|_F^2$. For product-unitary
channels $\mathcal V_{\mathrm{in}},\mathcal V_{\mathrm{out}}$ and any channel
$\mathcal E$, the PTM $R'$ of
$\mathcal V_{\mathrm{out}}\circ\mathcal E\circ\mathcal V_{\mathrm{in}}$
satisfies $b_{12}(R')=b_{12}(R)$.
\end{proposition}
\begin{proof}
Product unitaries preserve Pauli weight and act orthogonally on each weight
sector, so $R'_{21}=O_2^{\mathrm{out}}R_{21}O_1^{\mathrm{in}}$ and
$R'_{12}=O_1^{\mathrm{out}}R_{12}O_2^{\mathrm{in}}$.
\end{proof}

In Eq.~\eqref{eq:cancellation}, $b_{12}=2\sqrt2|\sin\lambda|$ for every
$\epsilon$. This full-block countermeasure is not evaluated at finite shots;
all detection results use $x$. For small rotations, a first-order expansion
with the weight-preserving orthogonal PTMs $O$ and $O'$ of
$\mathcal M_\epsilon$ and $\mathcal M_{\epsilon'}$ gives
$\|x\|_2\ge2(1-\eta-\eta')|\lambda|\,\|h\|_2-O(\lambda^2)$, where
$\eta=(\|O_1-I\|_2+\|O_2-I\|_2)/2\le\sin(\epsilon_{\max}/2)+\sin\epsilon_{\max}$
independently of $n$ and $\eta'$ is defined likewise, so the bound of
Section~\ref{sec:threat} prevents first-order cancellation.

\subsection{Feature images and normalization}
\label{sec:features}

The signed vector $x$ carries the security decision. For visualization and
for the quantum encoding of Section~\ref{sec:qhed} we use nonnegative
feature images derived from the PTM. The primary one is the cross-weight
magnitude image $S^{\mathrm{cross,anti}}$, obtained by placing
$|A_{\alpha\beta}|$ in both cross-weight orientations and zero elsewhere.
Three complementary images capture other channel properties:
\begin{align}
S^{\mathrm{anti,tr}}_{\alpha\beta}
&=
\tfrac12\left|R_{\alpha\beta}-R_{\beta\alpha}\right|,
\qquad
P_\alpha,P_\beta\neq I^{\otimes n},
\label{eq:antitr}
\\
S^{\mathrm{nonunital}}_{\alpha}
&=
|R_{\alpha,0}|,
\qquad
P_\alpha\neq I^{\otimes n},
\label{eq:nonunital}
\\
S^{\mathrm{cross,dev}}_{\alpha\beta}
&=
\left|R_{\alpha\beta}-\delta_{\alpha\beta}\right|,
\quad
\{\wt(P_\alpha),\wt(P_\beta)\}=\{1,2\}.
\label{eq:cross}
\end{align}
The antisymmetric-traceless image isolates coherent mixing over the whole
traceless sector; the nonunital image isolates identity-to-Pauli leakage
such as amplitude damping; the cross-weight deviation image collects all
weight mixing, coherent or not. We also use the total deviation
$S^{\mathrm{dev}}_{\alpha\beta}=|R_{\alpha\beta}-\delta_{\alpha\beta}|$ as a
magnitude reference.

The cross-weight masks are symmetric in the two orientations because
nonunital channels leak directionally: product amplitude damping turns a
spectator identity into a $Z$ component, so a weight-1 column acquires
weight-2 rows. A directed mask would report amplitude-damping cross mass
$0$, whereas the symmetric mask reports $\|S^{\mathrm{cross,dev}}\|_F=0.064$
at matched infidelity $0.02$; all results use the symmetric definition.

Each image is normalized in two stages: $\|S\|_F$ measures the amount of
structure, $\widehat S=S/\|S\|_F$ retains its pattern, and FRQI loading uses
$\widetilde S_j=S_j/\max_k S_k\in[0,1]$; the local unital null $S=0$ is
reported directly as the absence of a cross-weight image.

For continuity with matched-infidelity sweeps we also report the normalized
ratio
\begin{equation}
\rca
=
\frac{\|S^{\mathrm{cross,anti}}\|_F}
{\|S^{\mathrm{dev}}_{w=1}\|_F+\|S^{\mathrm{dev}}_{w=2}\|_F+\epsilon_0},
\label{eq:rca}
\end{equation}
where the denominator collects same-weight deviation mass and
$\epsilon_0>0$ guards the zero denominator. The security statements are
statements about $x$; the ratio is a reporting convention. For a pure attack
channel every first-order low-weight entry is cross-weight, so the
denominator is $O(\lambda^2)$ and the ratio admits the closed form
\begin{equation}
\rca^{\mathrm{pure}}(\lambda)
=
\frac{2\sqrt{2}\,\sin\lambda}{6\,(1-\cos\lambda)}
=
\frac{\sqrt{2}}{3}\cot(\lambda/2),
\label{eq:rca-exact}
\end{equation}
verified numerically to $10^{-6}$ over $\lambda\in[0.02,0.8]$ and equal to
$3.105$ at matched infidelity $0.02$.

\section{The Operational Detector}
\label{sec:detector}

The defender deploys a single algorithm on the signed vector
$x\in\mathbb{R}^{243}$ of the residual PTM.

\begin{enumerate}
    \item \emph{Estimate.} Acquire one shared randomized-Pauli
    classical-shadow dataset of $M$ total settings and form the estimate
    $\hat x$ from it. All entries read the same shots, so their estimators
    are correlated (Section~\ref{sec:shots}).

    \item \emph{Calibrate.} From an attack-free window learn per-coordinate
    null means $\mu_j$ and standard deviations $\sigma_j$, regularized as
    $\sigma_j\leftarrow\max(\sigma_j,\kappa\,\mathrm{median}_j\,\sigma_j)$.
    We fix $\kappa=0.25$ before attack evaluation in all experiments so that
    near-zero-variance coordinates cannot dominate the maximum.

    \item \emph{Score.} Compute the direction-resolved statistic
    \begin{equation}
    T
    =
    \max_j
    \frac{|\hat x_j-\mu_j|}{\sigma_j}.
    \label{eq:max-statistic}
    \end{equation}

    \item \emph{Threshold.} Set $\tau$ to the $(1-\alpha)$ quantile of $T$
    over a held-out attack-free window and flag if $T>\tau$.

    \item \emph{Corroborate.} Temporal-differential and context-contrast
    handles (Section~\ref{sec:handles}) serve as secondary confirmation.
\end{enumerate}

The max form is chosen over the plain norm $\|\hat x-\mu\|_2$ because the
isometry of Section~\ref{sec:isometry} makes cross-weight directions faithful
and non-cancelling, so a per-coordinate test can separate an attack
direction from a native-coupling direction. Calibrating $\tau$ on the
empirical distribution of $T$ absorbs the multiplicity of the maximum over
243 correlated coordinates (mean off-diagonal correlation $0.073$, up to
$0.38$): an independence-based threshold gives false-positive rate $0.15$ at
target $\alpha=0.05$, whereas the empirical quantile gives a held-out rate
of $0.013$ (95\% Wilson interval $0.004$--$0.047$) on $150$ independent null
records at $M=8{,}192$.

The detector must not fire when benign drift changes between calibration
and evaluation. For unital product noise with local drift, $x=0$ throughout
(Proposition~\ref{prop:drift-zero}); for nonunital product noise, the drift
PTM acts orthogonally on the cross-weight block, so $\|x\|_2$ is exactly
invariant while the vector rotates by $O(\epsilon)$. Numerically, the vector
changes by at most $10^{-16}$ over $\epsilon\in[0,0.3]$ in the unital case,
and under realistic noise the norm changes by $3\times10^{-17}$ while the
vector moves by about $10^{-4}$, so the calibration window should span the
expected drift range.

\section{Evaluation}
\label{sec:evaluation}

All experiments use $n=3$ qubits and the $k=2$ sector ($K=37$), simulated
with custom Python~3.10 code using NumPy and SciPy, without Qiskit or QuTiP:
PTMs are built from Kraus operators, pulses use time-ordered integration
with $400$ steps per gate, and FRQI/QHED is a $13$-qubit statevector.
Realistic noise is thermal relaxation with $T_1=50~\mu\mathrm{s}$ and
$T_2=70~\mu\mathrm{s}$ over a $200$-ns gate, depolarizing probability
$10^{-3}$, and amplitude damping, composed as
$\mathcal{E}=\mathcal{N}_{T_1,T_2,p}\circ\mathcal{A}_\lambda\circ
\mathcal{M}_\epsilon$. Unless noted, compared families are matched by
numerical bisection in average gate infidelity $r=1-F_{\mathrm{avg}}$, with
$F_{\mathrm{avg}}=(\Tr R+d)/(d(d+1))$ and $d=2^n$, and finite-shot studies
use fixed seeds.

The finite-shot studies follow a common protocol. Null records used to
estimate $\mu_j$ and $\sigma_j$ are disjoint from those used to select
thresholds, and reported false-positive rates use a further held-out set,
preventing optimistic reuse of the null samples. The detection-limit study
of Section~\ref{sec:shots} uses moment, threshold, and per-coupling attack
record counts that decrease from $240/240/140$ at $M=4{,}096$ to
$60/60/35$ at $M=16{,}384$; the held-out false-positive rate of the same
detector is reported in Section~\ref{sec:detector}. The native ZZ study of Section~\ref{sec:nativezz} uses
$200$ records each for moments, threshold selection, and held-out
reporting, with $120$ attack records per coupling. Every record is scored
from a fresh shared-shadow dataset of $M$ single-shot settings.

\subsection{Mechanism separation at matched infidelity}
\label{sec:matched}

We first verify that the feature images of Section~\ref{sec:features}
respond to the mechanism of an error rather than its magnitude.
Table~\ref{tab:test1} and Fig.~\ref{fig:test1} compare a local coherent
rotation, depolarizing noise, dephasing, and amplitude damping at average
gate infidelity $0.02$: the antisymmetric-traceless norm isolates the
coherent channel with a margin of $+1.065$, and the nonunital norm isolates
amplitude damping, neither being visible to the common scalar fidelity.

\begin{table}[tb]
\centering
\caption{Feature Frobenius norms at matched average gate infidelity $0.02$
for $n=3$, $k=2$. The coherent rotation dominates the
antisymmetric-traceless feature, while amplitude damping dominates the
nonunital and cross-weight deviation features.}
\label{tab:test1}
\setlength{\tabcolsep}{2.5pt}
\begin{tabular}{lcccc}
\toprule
Channel
& $\|S^{\mathrm{anti,tr}}\|_F$
& $\|S^{\mathrm{nonunital}}\|_F$
& $\|S^{\mathrm{dev}}\|_F$
& $\|S^{\mathrm{cross,dev}}\|_F$ \\
\midrule
$R_x$ rotation    & 1.1104 & 0.0000 & 1.1231 & 0.0000 \\
Depolarizing      & 0.0000 & 0.0000 & 0.1084 & 0.0000 \\
Dephasing         & 0.0000 & 0.0000 & 0.1225 & 0.0000 \\
Amp.\ damping & 0.0450 & 0.0262 & 0.1318 & 0.0636 \\
\bottomrule
\end{tabular}
\end{table}

\begin{figure}[tb]
\centering
\includegraphics[width=0.88\columnwidth]{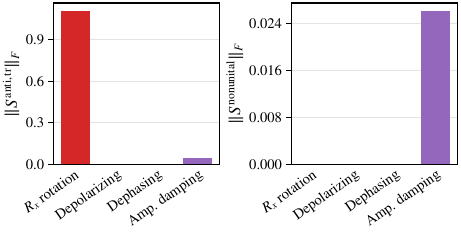}
\caption{Feature norms for four channel families at matched average gate
infidelity $0.02$. The antisymmetric-traceless feature separates coherent
rotation from the tested incoherent channels, while the nonunital feature
identifies amplitude damping.}
\label{fig:test1}
\end{figure}

\subsection{Drift versus attack: the structural separation (RQ1)}
\label{sec:separation}

The central experiment matches local coherent drift against ZX, YX, and ZZ
attacks at average gate infidelities from $0.005$ to $0.06$. In
Fig.~\ref{fig:discriminator} (right), the drift value of $\rca$ is
identically zero, as Proposition~\ref{prop:drift-zero} requires, while all
attacks follow Eq.~\eqref{eq:rca-exact}; at infidelity $0.02$ the attack
value is $3.105$, a zero-versus-nonzero distinction. This answers RQ1 in the
noiseless limit; the rest of this section addresses RQ3.

The figure also retains a failed baseline for contrast. The weight-sector
ratio $r_w$, the edge score of the weight-2 block over the weight-1 block,
fails because local drift on one qubit also produces weight-2 tensor copies:
at matched infidelity $0.02$ the drift value $r_w=13.0$ exceeds the attack
value $r_w=8.0$ (Fig.~\ref{fig:discriminator}, left). Scalar infidelity and
RB error rates do not locate the coherent component, $r_w$ confounds product
and interaction structure, the scalar cross-weight norm loses
directionality, and the direction-resolved statistic of
Eq.~\eqref{eq:max-statistic} retains the interaction geometry.

\begin{figure}[tb]
\centering
\includegraphics[width=\columnwidth]{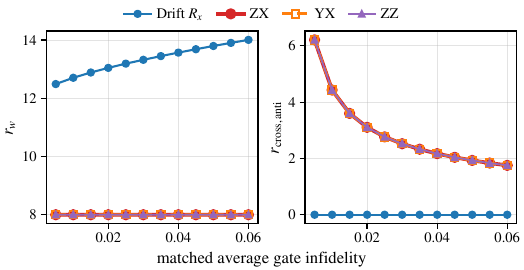}
\caption{Left: the weight-sector ratio $r_w$ fails; local drift exceeds all
attacks at matched infidelity because product drift creates weight-2
copies. Right: the corrected discriminator $\rca$ separates exactly; drift
is identically zero (Proposition~\ref{prop:drift-zero}) while all three
attack families follow
$\rca^{\mathrm{pure}}=\tfrac{\sqrt2}{3}\cot(\lambda/2)$ over the tested
range. This normalized ratio is used only to visualize mechanism
separation; it is not the operational detection statistic, which is the
calibrated test of Eq.~\eqref{eq:max-statistic} on the signed vector $x$.}
\label{fig:discriminator}
\end{figure}

Real devices are not products of unital channels, and amplitude damping
gives a small nonzero raw value ($\rca=0.016$ with no attack under the full
noise model). 

Baseline subtraction against a calibration estimate of the
same device restores the separation: the no-attack differential returns to
zero while ZX, YX, and ZZ attacks retain $\drca\approx1.05$ at infidelity
$0.02$. Under shadow noise at $65{,}536$ settings (200 trials), the
separation stays clean, with attack minimum $0.88$ against a drift and
noise holdout maximum of $0.17$.

\subsection{Temporal and context handles}
\label{sec:handles}

  Two secondary handles, specializations of the context and drift tests
  of \cite{rudinger2019context}, corroborate an alarm.
  The temporal differential
scores $\|S_t-\bar S\|_F$ against a calibration-window mean: slow drift moves
it smoothly, while attack insertion produces a jump. In a 12-frame timeline
with an attack inserted at frame 8, the first attack frame scores $0.91$
against drift frames bounded by $0.30$ (Fig.~\ref{fig:temporal}).

The context contrast compares residual features with the spectator idle
versus driven. Victim-local drift commutes with a spectator-only drive, so
its contrast is exactly zero, while two-qubit attack terms do not commute and
produce a contrast that grows with infidelity (Fig.~\ref{fig:context}). The
handle is scoped to victim-local drift; spectator-local drift is already
removed by the cross-weight signature.

\begin{figure}[tb]
\centering
\includegraphics[width=0.9\columnwidth]{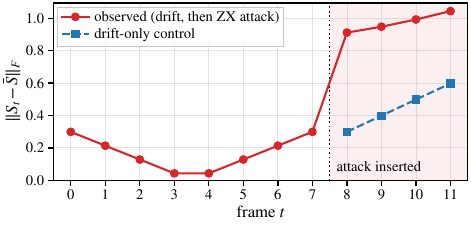}
\caption{Temporal differential over a 12-frame timeline with attack
insertion at frame 8. Drift frames stay below $0.30$; the first attack
frame scores $0.91$.}
\label{fig:temporal}
\end{figure}

\begin{figure}[tb]
\centering
\includegraphics[width=0.85\columnwidth]{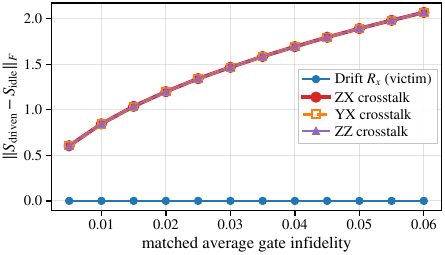}
\caption{Context contrast between driven and idle spectator conditions
across matched infidelities. Victim-local drift cancels exactly; all three
crosstalk families produce nonzero contrast at every tested infidelity.}
\label{fig:context}
\end{figure}

\subsection{Finite measurements and detection limits (RQ3)}
\label{sec:shots}

The low-weight block is estimated from a single shared dataset, as on
hardware and in randomized product-Pauli characterization \cite{crupi2025lowdegree}. Each setting draws one random product-Pauli input eigenstate and
one random measurement basis, is executed with a single shot, and
contributes to every PTM entry it covers; $R_{\alpha\beta}$ is covered when
the bases match the supports of $P_\alpha$ and $P_\beta$, and the two-sided
inverse-shadow factor $3^{w_\alpha+w_\beta}$ cancels the match probability
in expectation. Budgets $M$ count total settings, equal here to total shots,
and measure sample complexity rather than audit latency; the per-entry
cross-weight variance scales as $27/M$, so the noise floor on $\hat x$
scales as $M^{-1/2}$.

Figure~\ref{fig:detlimit} reports the resulting detection limit
$\lambda_{\min}$ of the detector of Section~\ref{sec:detector}, the
smallest coupling reaching true-positive rate at least $0.95$ at the
threshold calibrated for $\alpha=0.05$, as a function of $M$.
The fitted exponents are $-0.51$ for ZX, $-0.56$ for YX, and $-0.52$ for
ZZ, all consistent with the predicted $M^{-1/2}$ scaling, and the three
families track one another, as expected from the isometry of
Lemma~\ref{lem:isometry}. At
$M=16{,}384$ total settings the limit is $\lambda_{\min}\approx0.13$ for
all families under the two-sided accounting. The tested couplings increase
by a factor of about $1.5$, and $\lambda_{\min}$ is log-interpolated
between the two that bracket the $0.95$ crossing, so the limits and
exponents are resolved only to that bracket.

\begin{figure}[tb]
\centering
\includegraphics[width=\columnwidth]{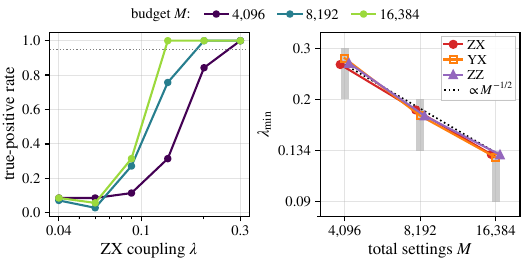}
\caption{Detection limit of the calibrated maximum-$z$ detector under
two-sided shared-dataset shadow estimation. Left: ZX true-positive rate
versus coupling for increasing total budgets $M$. Right: $\lambda_{\min}$
versus $M$ for all three families (markers offset horizontally for
visibility) with the $M^{-1/2}$ guide; gray bars mark the adjacent tested
couplings that bracket the $0.95$ crossing.}
\label{fig:detlimit}
\end{figure}

\subsection{The native ZZ baseline (RQ3)}
\label{sec:nativezz}

Motivated by transmon characterization \cite{oda2026sparse}, the
calibration channel carries an always-on ZZ coupling
$\zeta_{\mathrm{cal}}=0.10$ rad per $200$-ns gate that fluctuates as
$\zeta_{\mathrm{eval}}=\zeta_{\mathrm{cal}}+\delta\zeta$,
$\delta\zeta\sim\mathcal{N}(0,\sigma_\zeta^2)$. The null is native wander
only; the attack adds excess coupling on the same pair. All channels pass
through the two-sided estimator at $65{,}536$ settings, and we compare the
scalar norm with the direction-resolved statistic of
Eq.~\eqref{eq:max-statistic}.

Figure~\ref{fig:nativezz} shows that the resulting hardware floor is
directional, as the geometry of Section~\ref{sec:isometry} predicts. Excess
ZZ points along the native-wander direction, so its limit rises with
instability regardless of statistic, from the shot floor $0.030$ at
$\sigma_\zeta=0.002$ to about $0.080\approx4\sigma_\zeta$ at
$\sigma_\zeta=0.02$. Attacks along other directions populate cross-weight
coordinates that native wander never touches: with the direction-resolved
statistic, ZX and YX hold $\lambda_{\min}\approx0.05$ to $0.08$ essentially
independent of $\sigma_\zeta$, while the scalar norm degrades by about a
factor of two. Tracking the native coupling \cite{oda2026sparse} recovers
part of the collinear component.

\begin{figure}[tb]
\centering
\includegraphics[width=\columnwidth]{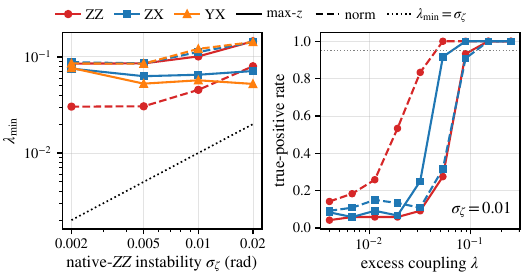}
\caption{Calibrated native ZZ baseline ($\zeta_{\mathrm{cal}}=0.10$,
two-sided shadows, $65{,}536$ settings). Left: $\lambda_{\min}$ versus
instability $\sigma_\zeta$ for excess ZZ, ZX, and YX under the scalar norm
and the direction-resolved statistic. Right: true-positive rate at
$\sigma_\zeta=0.01$. Excess ZZ is floor-limited; ZX and YX under the
direction-resolved statistic stay at the shot floor.}
\label{fig:nativezz}
\end{figure}

Real calibration drift can also be heavy-tailed, temporally correlated, or
dominated by discrete two-level-fluctuator jumps. Repeating the experiment
with the null drawn from a Student-$t$ distribution (three degrees of
freedom), an AR(1) process ($\rho=0.8$), and symmetric telegraph jumps, all
matched to $\sigma_\zeta=0.04$ at $M=16{,}384$, preserves the directional
advantage: off-collinear ZX limits stay in $[0.11,0.16]$ near the shot
floor, while collinear ZZ limits remain instability-limited in
$[0.18,0.30]$. The hardest case is temporally correlated drift, which
narrows the gap and pushes the held-out false-positive rate to $0.10$
against a target of $0.05$, because a colored null is poorly summarized by
per-coordinate moments; strongly autocorrelated nulls require block
calibration, an effective sample-size correction, or a time-series-aware
threshold.

\subsection{Pulse-level validation (RQ3)}
\label{sec:pulse}

Constant-generator unitaries are algebraic stand-ins for real drives, so
we validate the detector on time-dependent models. A Hamiltonian residual
model factors out an intended drive
$H_{\mathrm{id}}=\tfrac{\pi}{4}X_v+\tfrac{\pi}{6}X_s$ before scoring, with
realistic noise; ZX, YX, and ZZ terms give $\drca=0.110$, $0.109$, and
$0.097$ at $0.02$ rad, growing monotonically to $1.61$, $1.61$, and $1.41$
at $0.30$ rad. A cross-resonance pulse whose terms share one Gaussian-square
envelope reproduces these values to three figures, as it must, since a
common envelope commutes with itself; this checks the pipeline, and the
integrator matches the exact exponential to $2.6\times10^{-14}$.

The third model makes time ordering nontrivial: the attack term carries its
own detuned, phase-modulated envelope,
$H(t)=\Omega(t)H_{\mathrm{drive}}
+\tfrac{j}{2}\Omega(t)\cos(\Delta t+\varphi)\,G$, with
$\Delta t_{\mathrm{gate}}\in\{\pi,2\pi,4\pi\}$ and $\varphi=\pi/4$. The
witness $\|U_{\mathrm{pulse}}-e^{-i\int H}\|$ reaches $0.06$, against
$10^{-14}$ for the common envelope, so the evolution departs from the
un-ordered integrated Hamiltonian. Two results follow
(Fig.~\ref{fig:noncomm}). First, for drift-varied nulls with
$\epsilon_{\mathrm{eval}}\in[0.04,0.12]$ around the calibration value
$0.08$, the numerator differential $\Delta\|S^{\mathrm{cross,anti}}\|_F$
stays at $4\times10^{-16}$, the empirical face of the drift invariance of
Section~\ref{sec:detector}, while the ratio moves by $0.028$ through its
denominator. Second, detection persists at every detuning: the envelope
overlap $\chi(\Delta)=|\int\Omega\cos(\Delta t+\varphi)\,dt|/\int\Omega\,dt$
falls from $0.71$ to $0.07$ at $\Delta t_{\mathrm{gate}}=4\pi$, and all
curves remain monotone in $j$ and clear the drift-null band at the smallest
tested coupling.

\begin{figure}[tb]
\centering
\includegraphics[width=\columnwidth]{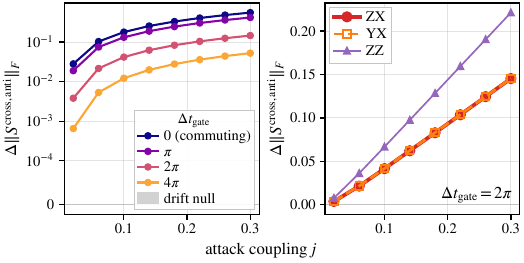}
\caption{Time-ordering-sensitive pulse validation. Left: numerator
differential versus bare ZX coupling $j$ for detunings
$\Delta t_{\mathrm{gate}}\in\{0,\pi,2\pi,4\pi\}$; the shaded band is the
drift-varied null, which is machine zero for the numerator. Right: all
three attack families at $\Delta t_{\mathrm{gate}}=2\pi$. Detection
persists under genuinely non-commuting evolution.}
\label{fig:noncomm}
\end{figure}

\subsection{Resource scaling}
\label{sec:scaling}

Table~\ref{tab:scaling} summarizes growth at fixed $k=2$, with counts for
$n\le4$ taken from the exact graph builder and larger $n$ from closed forms
validated against it. The monitored objects grow polynomially:
$K=\Theta(n^2)$ low-weight Paulis, $\Theta(n^3)$ cross-weight coordinates,
and $\Theta(n^4)$ generic FRQI controlled rotations, with an address
register of only $\Theta(\log n)$ qubits. Each cross-weight coordinate has
a two-sided shadow variance factor of at most $3^{2k}=81$, constant at
fixed $k$. Because one shared dataset estimates all coordinates
simultaneously, the total budget is not a per-entry cost multiplied by the
entry count; the entry count controls classical storage and
post-processing.

Beyond three qubits, commuting Ising residuals
$U=\exp[-i\sum_{(i,j)\in E}\theta_{ij}Z_iZ_j/2]$ admit the exact response
\begin{equation}
R_{Y_iZ_j,X_i}=\sin\theta_{ij}\prod_{k\in N(i)\setminus\{j\}}\cos\theta_{ik},
\label{eq:ising}
\end{equation}
where $N(i)$ denotes the neighbors of qubit $i$. Thus cross-weight
visibility depends on interaction angles and neighborhood degree. Pairwise
full-block norms $b_{ij}$, restricted to weight-1 and weight-2 Paulis
supported on $\{i,j\}$, retain invariance under local-unitary pre- and
postcomposition because product unitaries preserve Pauli support, providing
features for graph-based analysis of distributed coupling anomalies. FRQI/QHED can implement the corresponding graph readout;
larger-system finite-shot performance, localization accuracy, and
end-to-end quantum advantage remain to be evaluated.

\begin{table}[tb]
\centering
\caption{Resource scaling at fixed $k=2$. Counts for $n\le4$ are exact;
$n\ge5$ from closed forms validated at $n\le4$.}
\label{tab:scaling}
\setlength{\tabcolsep}{4pt}
\begin{tabular}{rccccc}
\toprule
$n$ & $K$ & Cross entries & Hamming edges & FRQI qubits & FRQI rot. \\
\midrule
3 & 37  & 243    & 9\,324   & 11 & 1\,369 \\
4 & 67  & 648    & 32\,160  & 13 & 4\,489 \\
5 & 106 & 1\,350 & 82\,680  & 14 & 11\,236 \\
6 & 154 & 2\,430 & 177\,408 & 15 & 23\,716 \\
\bottomrule
\end{tabular}
\end{table}

\section{Quantum-Native Pauli-Graph Edge Readout}
\label{sec:qhed}

The detector operates on the signed vector $x$. RQ4 asks whether the
PTM-derived image admits a quantum-native representation and readout; this
section constructs one and states what it does and does not provide.

\subsection{FRQI encoding of the PTM feature}

Let $\widetilde S_j\in[0,1]$ denote the normalized nonnegative feature
values, padded with zeros to $N_{\mathrm{pad}}=2^m$ entries. The FRQI state
is
\begin{equation}
|\mathrm{FRQI}(S)\rangle
=
\frac{1}{\sqrt{N_{\mathrm{pad}}}}
\sum_{j=0}^{N_{\mathrm{pad}}-1}
|j\rangle
\left(
\sqrt{1-\widetilde S_j^2}\,|0\rangle
+
\widetilde S_j|1\rangle
\right),
\label{eq:frqi-state}
\end{equation}
prepared by controlled rotations with angles
$\theta_j=\arcsin\widetilde S_j$ \cite{le2011frqi}; the address register
identifies a PTM coordinate and the final qubit stores its magnitude.
Postselecting the final qubit on $|1\rangle$ yields
$|\psi_S\rangle=\|S\|_2^{-1}\sum_j S_j|j\rangle$ with success probability
$\sum_j\widetilde S_j^2/N_{\mathrm{pad}}$, which vanishes exactly for the
local unital null. The state keeps magnitude and support but not sign, so
$x$ remains the detection object.

\subsection{Pauli graph and QHED score}

A PTM pixel is indexed by a pair of Pauli strings $(P_\alpha,P_\beta)$, and
lexicographic adjacency between such pairs is physically arbitrary, so the
cyclic-shift neighborhood of standard QHED \cite{yao2017qhed} does not
apply. Our Pauli graph $G=(V,E)$ connects PTM coordinates in which one label
is unchanged and the other differs in one tensor factor; a sparser variant
also requires the differing factors to anticommute. At $n=3$, $k=2$ the two
graphs have $9{,}324$ and $4{,}662$ edges in $18$ and $9$ matchings
(Fig.~\ref{fig:pauligraph}).

\begin{figure}[tb]
\centering
\includegraphics[width=0.92\columnwidth]{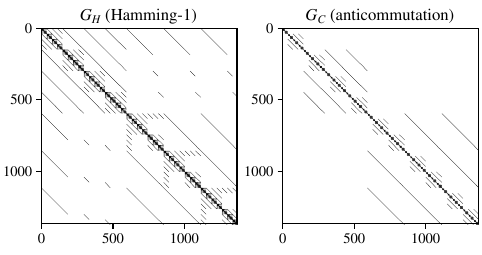}
\caption{Adjacency of the two Pauli graphs on the $K\times K$ low-weight
block at $n=3$, $k=2$ ($K=37$). Left: Hamming-1 graph, $9{,}324$ edges.
Right: anticommutation graph, $4{,}662$ edges.}
\label{fig:pauligraph}
\end{figure}

The normalized Pauli-graph edge score is the graph Dirichlet energy
\begin{equation}
\Aedge(S)
=
\frac{1}{4\sum_j S_j^2}
\sum_{(u,v)\in E}
|S_u-S_v|^2
=
\frac{S^TL_GS}{4\,S^TS},
\label{eq:qhed-edge-score}
\end{equation}
where $L_G$ is the graph Laplacian. QHED implements it through the
matching decomposition: each matching $M_\ell$ defines an involutive
permutation $D_{M_\ell}$ of the address register, and a Hadamard on an
auxiliary qubit, a controlled $D_{M_\ell}$, and a second Hadamard yield
$P_\ell=\tfrac14\sum_j|a_j-a_{D_{M_\ell}(j)}|^2$ for $a_j=S_j/\|S\|_2$, so
$\Aedge(S)=\tfrac12\sum_\ell P_\ell$.

The detector and the QHED score answer different questions: $\|x\|_2$
measures the amount of cross-weight coherent structure, while $\Aedge(S)$
measures how that structure is arranged over the Pauli graph.

\subsection{Circuit verification and ablation (RQ4)}

A gate-level statevector implementation with $11$ address qubits, one
pixel qubit, and one auxiliary qubit executes FRQI loading, postselection,
and the edge primitive on every matching of both graphs, and reproduces the
classical edge score for all tested channels, features, and graphs with
maximum error $1.4\times10^{-14}$ (Fig.~\ref{fig:frqiqhed}), which answers
RQ4 at the level of exact realization.

\begin{figure}[tb]
\centering
\includegraphics[width=0.72\columnwidth]{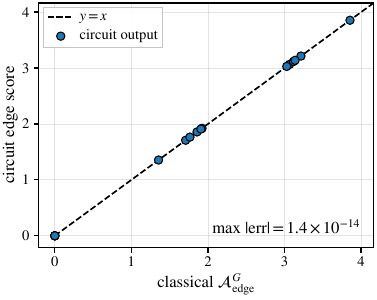}
\caption{Circuit-level verification of the Pauli-graph QHED readout: the
statevector circuit matches the classical edge score across all tested
channels, features, and both graphs (maximum error $1.4\times10^{-14}$).}
\label{fig:frqiqhed}
\end{figure}

Two qualifications bound the claim. First, the implemented route is
classically compiled: the PTM is estimated classically, the FRQI rotations
are compiled from those values, and the graph score is classically
computable in $O(|E|)$ time once $S$ is known, so no quantum speedup is
claimed; a fully coherent route needs a PTM-derived value oracle or
equivalent block encoding. Second, QHED is not the strongest detector: in
the drift-versus-ZX ablation of Table~\ref{tab:ablation}, the structural
statistics dominate at every coupling, while QHED remains the quantum-native
readout of graph morphology.

FRQI/QHED provides a quantum readout of graph-structured coupling features.
Scaling benefits could arise when coherent feature preparation avoids
explicit classical reconstruction, or when reversible access enables
amplitude amplification and estimation; an end-to-end advantage must
account for feature acquisition, state preparation, graph operations, and
competing classical sampling methods. For a unitary residual $U$ with $U$
and $U^{\dagger}$ applicable to the monitored qubits, or equivalently $U$
and $U^{*}$ on the halves of a Bell register $|\Phi\rangle$, Pauli-Bell
encoding of $(UP_\beta U^{\dagger}\otimes I)|\Phi\rangle$ over the $3n$
weight-1 inputs, Bell decoding, and postselection on weight-2 outputs
prepare the normalized signed $R_{21}$ state
$\propto\sum_{\alpha,\beta}R_{\alpha\beta}|\alpha\rangle|\beta\rangle$ with
success probability $\|R_{21}\|_F^2/(3n)$. For pure Ising dynamics on a
bipartite coupling graph, $X$ echo pulses on one partition implement
$U^{\dagger}=U^{*}$. This gives a concrete route to QHED without classical
feature loading; robustness under noisy dynamics and an end-to-end advantage
remain to be established.

\begin{table}[tb]
\centering
\caption{Ablation: drift-versus-ZX detection ROC-AUC under shadow noise.
Structural statistics dominate; QHED serves as the quantum-native
morphology readout.}
\label{tab:ablation}
\begin{tabular}{ccccc}
\toprule
$\lambda$ & $\|x\|_2$ & max-$z$ & graph-TV & QHED edge \\
\midrule
0.05 & 0.599 & 0.542 & 0.571 & 0.551 \\
0.10 & 0.748 & \textbf{0.787} & 0.679 & 0.662 \\
0.15 & 0.936 & \textbf{0.979} & 0.852 & 0.882 \\
\bottomrule
\end{tabular}
\end{table}

\section{Limitations and Countermeasures}
\label{sec:limitations}

\paragraph{Local-rotation evasion}
The evaluated statistic vanishes for an attack composed with a local $\pi$
rotation (Section~\ref{sec:rotation}); the full-block norm removes this
cancellation but still requires finite-shot calibration.

\paragraph{Incoherent attacks}
A stochastic correlated injection, such as random-sign ZZ or correlated
dephasing, produces a Pauli-diagonal PTM with zero cross-weight mass and
evades the discriminator, a boundary inherent to any coherence-based
detector. Magnitude features and the temporal handle give partial coverage,
and purity-based or randomized benchmarking diagnostics target this class.

\paragraph{Weight truncation}
Coherent generators of weight greater than $k$ escape the monitored
sector. Raising $k$ trades polynomially more estimation cost for coverage.

\paragraph{Native-crosstalk mimicry}
Excess ZZ within the native fluctuation band is indistinguishable from
calibration wander, with $\lambda_{\min}\approx4\sigma_\zeta$. The floor is
directional, so non-collinear attacks remain detectable at the shot floor,
and tracking the native coupling \cite{oda2026sparse} tightens the
collinear floor.

\paragraph{Baseline poisoning}
All differential handles assume an attack-free calibration window. An
adversary active during calibration is absorbed into the baseline.
Randomizing calibration scheduling mitigates but does not eliminate this.

\paragraph{Diagnostic-window evasion}
A context-aware adversary who attacks only tenant workloads and stays
silent during audits is not observed. Provider-side deployment should
interleave audits randomly within tenant execution so that audit and
workload gates are statistically indistinguishable to the attacker; we do
not analyze the resulting game, the most important open adversarial
question.

\section{Conclusion}
\label{sec:conclusion}

This work establishes a structural approach to verifying multi-tenant
quantum execution beyond scalar error metrics. For RQ1, products of local
unital errors yield an exactly zero cross-weight PTM block, while every
nonzero two-qubit interaction contributes at first order. For RQ2, the
map from a specified pair's nine interaction coefficients to the
antisymmetric feature is an isometry up to scale for every system size,
ensuring sensitivity to all interaction directions near the identity.
The full cross-weight norm preserves this signal under arbitrary local
unitary composition, extending the structural guarantee beyond the
evaluated detector's small-rotation regime. For RQ3, the calibrated detector uses one shared randomized-Pauli dataset
to identify ZX, YX, and ZZ interactions, reaching
$\lambda_{\min}\approx0.13$ at $16{,}384$ settings with a $5\%$ calibrated
false-positive target and approximately $M^{-1/2}$ scaling. Experiments
support detection under relaxation, depolarization, and pulse dynamics,
while native ZZ fluctuations reveal directional detection floors and
colored drift motivates improved temporal calibration. For RQ4, QHED on
the FRQI-encoded feature reproduces the classical Pauli-graph edge score
within $1.4\times10^{-14}$; coherent feature preparation remains necessary
for an inline quantum diagnostic.

Together, these results connect exact structural guarantees to
measurement-based detection. Operator-indexed features derived from
basis algebra offer a broader route to quantum verification, with
device validation using measured drift and injected pulse perturbations
as the next step.

\newpage
\balance
\bibliographystyle{IEEEtran}
\bibliography{references}

\end{document}